\documentclass[journal]{IEEEtran}
\ifCLASSINFOpdf
\else
\fi
\usepackage{graphicx}
\usepackage{amssymb}
\usepackage{amsmath}
\usepackage{multicol}
\usepackage{stfloats}
\usepackage{cite}
\usepackage{booktabs}
\usepackage{makecell}
\newtheorem{theorem}{Theorem}

\newtheorem{lemma}{Lemma}

\newtheorem{assumption}{Assumption}
\newtheorem{remark}{Remark}
\newtheorem{proof}{Proof}
\newtheorem{problem}{Problem}

\usepackage{tikz,xcolor,hyperref}
\definecolor{lime}{HTML}{A6CE39}
\DeclareRobustCommand{\orcidicon}{%
	\begin{tikzpicture}
		\draw[lime, fill=lime] (0,0) 
		circle [radius=0.16] 
		node[white] {{\fontfamily{qag}\selectfont \tiny ID}}; 
		\draw[white, fill=white] (-0.0625,0.095) 
		circle [radius=0.007];	  
	\end{tikzpicture}
	\hspace{-2mm}}
\foreach \x in {A, ..., Z}{%
	\expandafter\xdef\csname orcid\x\endcsname{\noexpand\href{https://orcid.org/\csname orcidauthor\x\endcsname}{\noexpand\orcidicon}}
}
\makeatother

\begin{document}
\title{Distributed Multiconsensus Control of BESSs Based on Centrality of Eigenvectors}
\author{Yalin Zhang\orcidA{},
	Zhongxin Liu\orcidB{},~\IEEEmembership{Member,~IEEE,}
	Zengqiang Chen\orcidC{}
	\thanks{Manuscript received XX, XX; revised XX, XX;
		accepted XX, XX. This work is supported in part by the National Natural Science Foundation of China (Grant No. 92367105) and the General Terminal IC Interdisciplinary Science Center of Nankai University (Corresponding author: Zhongxin Liu.). 
		\par The authors are with the College of Artificial Intelligence, Nankai University, Tianjin 300350, and also with the Tianjin Key Laboratory of Interventional Brain-Computer Interface and Intelligent Rehabilitation, Nankai University, Tianjin 300350, China (e-mail: zhangyl@mail.nank.edu.cn; lzhx@nankai.edu.cn; chenzq@nankai.edu.cn).
}}
\markboth{IEEE TRANSACTIONS ON CIRCUITS AND SYSTEMS II:EXPRESS BRIEFS, VOL. XX, NO. XX, XX XX}%
{Zhang \MakeLowercase{\textit{et al.}}: Distributed Multiconsensus Control of BESSs Based on Centrality of Eigenvectors}
\maketitle
\begin{abstract}
	Secondary control and the State-of-Charge (SoC) balance control are important control objectives for battery energy storage systems (BESSs). In this brief, a communication weight allocation method based on the centrality of eigenvectors is designed for a connected and directed graph, which results in the adjacency matrix having a given eigenvector. Subsequently, a distributed secondary voltage controller and an SoC balancing controller are designed for droop-controlled BESSs to achieve voltage leader-following multiconsensus and SoC balancing, respectively. It is worth mentioning that under the designed voltage secondary control scheme, only a single leader is needed to achieve voltage multiconsensus control. In addition, the capacity information/droop coefficient does not need to be transmitted in the communication network to achieve power sharing according to capacity and SoC balance. For SoC balance control, the control gain is also well analyzed to ensure stability. The relevant simulations verify the effectiveness of the designed scheme.
\end{abstract}
\IEEEpeerreviewmaketitle
\section{Introduction}
\IEEEPARstart{B}{attery} energy storage systems (BESSs) are widely equipped in smart grids due to their role in peak shaving and valley filling \cite{9815324,HOSSAINLIPU2022132188}. A BESS can serve as a backup power source and be activated in special situations such as insufficient energy supply or malfunctions of the power generation equipment \cite{9815324}. In addition, BESSs play an important in regulating bus voltage and frequency \cite{HOSSAINLIPU2022132188}.
\par In existing distributed solutions, it is a common objective for microgrids to restore the output voltage of each BESS to the rated value of the utility grid. With this as the goal, some distributed secondary control schemes with finite time stability \cite{9531546}, resistance to communication time-delay \cite{QIN2024110255}, network attacks \cite{MO202448,9989432}, and disturbances \cite{9676705,10092457}, event triggering mechanism \cite{9531546,10032563}, self triggering \cite{9989432,LIU2023108679} and based on optimal control \cite{9733168} have been developed. Without exception, these are distributed solutions designed based on multi-agent systems \cite{9968139} and the synchronization theory in complex networks \cite{10614883}, and by them, the secondary control for BESSs avoids expensive centralized control centers and improves its robustness. However, this seem very idealistic and without considering line impedance because if the output voltage of each BESS is the same as the voltage at PCC, power flow does not exist. It seems that there is still a lot of room for improvement in the secondary voltage control when considering line impedance. 
\par In addition, for BESSs, SoC balance \cite{10600085,Zhang2020} and power sharing \cite{GAMAGE2024111180,Zhang2020} are important control objectives. Therefore, many researchers have developed asymptotically stable controllers \cite{10600085,Zhang2020}, prescribed time controller \cite{10527370} to accelerate convergence rate. However, to solve the problem, in previous solutions, power, SoC, and capacity information are packaged and sent to their neighboring agents to design controllers, which still seem to be able to be further expanded on information privacy.
\par In light of these, two distributed multiconsensus control schemes are developed in this article to improve existing secondary control schemes and SoC balance schemes. The specific contributions are as follows:
\begin{enumerate}	
	\item A communication weight allocation scheme for a connected and directed graph is designed to ensure that the adjacency matrix has a given leading eigenvector, which allows the value of eigenvector centrality to be pre-set.
	\item Based on the above mechanism, a leader-follower secondary voltage controller and a SoC balance controller are designed for droop controlled BESSs. Unlike previous schemes \cite{9531546,QIN2024110255,MO202448,9989432,9676705,10092457,10032563,LIU2023108679,9733168,10266581}, under these controllers, the steady-state of each voltage secondary control can be preset, and the capacity information/droop coefficient no longer needs to be transmitted in the communication network.
\end{enumerate}
\par Some preliminaries are arranged in Section II. Next, a distributed secondary control scheme for BESSs and a distributed scheme that achieves SoC balance and power sharing are developed in Section III. In order to verify the effectiveness of the designed schemes, two simulation examples are arranged in Section IV. Finally, a conclusion is drawn in Section V.
\section{Preliminaries}

\subsection{Graph Theory}
In this article, MASs are configured to manage BESSs, where this relationship is one-to-one. The communication between agents can be characterized by a digraph $\mathcal{G}(\mathcal{V},\mathcal{E})$, where $\mathcal{V}$ is the set of nodes, naming agents, in the graph and $\mathcal{E}$ represents the set of edges, naming communication links. An adjacency matrix ${\cal A}=[\mathrm{z}_{ij}]$ is used to characterize the communication relationship between them. That is, if agent $i$ can access agent $j$, $\mathrm{z}_{ij}=1$, agent $j$ is viewed as a neighbor of agent $i$, and $(\mathrm{v}_j,\mathrm{v}_i)\in\mathcal{E}$; otherwise, $\mathrm{z}_{ij}=0$.
\par For the leader-following mode, a matrix ${\cal B}=\mathrm{diag\{\mathrm{b}_{11},\cdots,\mathrm{b}_{\mathrm{nn}}\}}$ is defined as follows. If agent $i$ can access the leader agent, $\mathrm{b}_{ii}= 1$; $\mathrm{b}_{ii}=0$ otherwise.
\begin{assumption}
	\label{assu1}
	The communication topology utilized in this article is directed and connected.
\end{assumption}
\subsection{Matrix Theory}
\par In this subsection, some classic matrix theories are introduced for subsequent use.
\begin{theorem}
	\label{THPF}
	\cite{VVG2017} Let $\rho({\cal A})$ be the spectral radius of the matrix ${\cal A}$, where ${\cal A}$ is non negative and irreducible. If $\rho({\cal A})$ is an eigenvalue of ${\cal A}$, it is a simple one of ${\cal A}$ and associated with a eigenvector with all components being positive.
\end{theorem}
\begin{lemma}
	\label{lemeq}
	\cite{DHWAW1951} Let ${\cal A}$ and $\mathrm{b}$ be an $\mathrm{n}$-dimensional square matrix and an $\mathrm{n}$-dimensional column vector, respectively. One can find an $\mathrm{n}$-dimensional column vector $x$ or $y$ that makes only one of the following valid:
	\begin{enumerate}
		\item ${\cal A}x=\mathrm{b}$ with $x>0$,
		\item ${\cal A}^\mathrm{T}y\ge 0$ and $y^\mathrm{T}\mathrm{b}\ge 0$.
	\end{enumerate}
\end{lemma}
\begin{lemma}
	\label{nons}
	\cite{BIERKENS2014191} Assume that matrix $\cal M$ is singular and has a simple zero eigenvalue, $\nu$ and $\omega$ are the left and right eigenvectors corresponding to that eigenvalue, respectively. If there are two vectors $x$ and $y$ that make $(\nu^Tx)(y^T\omega)\neq0$ true, then ${\cal M}+xy^T$ is a non singular matrix.
\end{lemma}
\begin{lemma}
	\label{Mm}
	\cite{BIERKENS2014191} The result of Lemma \ref{nons} can be further deduced that ${\cal M}+xy^T>0$ if ${\cal M}$ is irreducible and a M-matrix and $\nu>0$, $\omega>0$, $x\le0$, and $y\le 0$.
\end{lemma}
\begin{lemma}
	\label{MH}
	\cite{doi:10.1137/20M131802X} A M-matrix ${\cal M}$ is considered to be Hurwitz if and only if ${\cal M}<0$.
\end{lemma}
\section{Design of Distributed Multiconsensus Controller of BESSs}
\begin{figure}
	\centering
	\includegraphics[width=8cm]{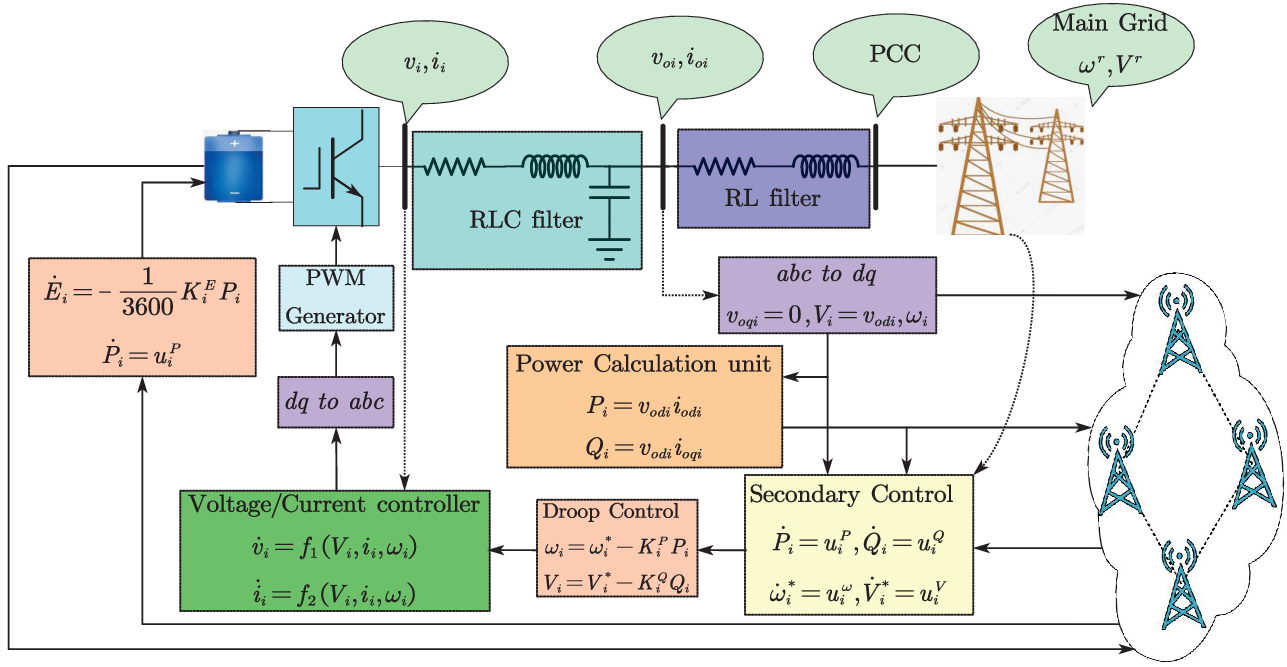}
	\caption{A block diagram for distributed control of a drop-controlled BESS}
	\label{AC}
\end{figure}
The hierarchical control mode is widely used in BESSs, as shown in Fig. \ref{AC}. A common practice is to ignore the voltage loop and current loop \cite{9531546,Zhang2020}, where droop control method is adopted as the primary control scheme. Thus, a battery indexed by $i$ is subjected to the following dynamics \cite{9815324,10600085,Zhang2020,GAMAGE2024111180}
\begin{subequations}
	\begin{equation}
		\dot{E}_{i}=-\frac{1}{\mathrm{T}}\mathrm{K}_{i}^{P}{P}_{i},
	\end{equation}
	\begin{equation}
		\mathrm{K}_{i}^{P}\dot{P}_{i}=u_{i}^{P},
	\end{equation}
\end{subequations}
where $E_i$ and $P_i$ are the SoC and the output power of BESS $i$, $\mathrm{K}_i^P$ denotes the capacity of BESS $i$, $\mathrm{T}$ is the time scale, $u_{i}^{P}$ is the control input.
\par An inverter-interfaced BESS driven by voltage droop control can be described by a simplified dynamic model as follows,
\begin{subequations}
	\begin{equation}
		V_i=V_{\mathrm{n},i}-\mathrm{K}_i^QQ_i,
	\end{equation}
	\begin{equation}
		\label{Qu}
		\dot{Q}_{i}=u_i^Q,
	\end{equation}
	\begin{equation}
		\label{Vn}
		\dot{V}_{\mathrm{n},i}=u_i^V+\mathrm{K}_i^Qu_i^Q,
	\end{equation}
\end{subequations}
where $V_{\mathrm{n},i}$, $V_i$, and $Q_i$ are the nominal and output voltage, reactive power, respectively, $\mathrm{K}_i^Q$, $u_i^V$, and $u_i^Q$ are a droop coefficient and the control inputs, respectively.
\par A communication weight $\Omega$ can be found by Lemma \ref{lemwt} such that a given vector $\mathrm{v}_1^{\prime}$ is the leading eigenvector of ${\cal A}\circ \Omega$, where $\circ$ denotes Hadamard product.
\begin{lemma}
	\label{lemwt}
	One can find a communication weight matrix $\Omega$ so that ${\cal A}\circ \Omega$ has the leading eigenvector $\mathrm{v}_1^{\prime}$ associated with $\rho({\cal A})$, where the elements of the matrix $\Omega$ are all non negative.
\end{lemma}
\begin{proof}
	The sufficient and necessary condition for $\mathrm{v}_1^{\prime}$ to be the leading eigenvector associated with $\rho({\cal A})$ of ${\cal A}\circ \Omega$ is
	\begin{equation}
		({\cal A}\circ \Omega)\mathrm{v}_1^{\prime}=\rho({\cal A})\mathrm{v}_1^{\prime}.
	\end{equation}
	By the vectorization defined in \cite{TOMASELLI2024111638}, one can obtain that
	\begin{equation}
		\label{M1}
		({\mathrm{v}_1^{\prime}}^\mathrm{T}\otimes \mathrm{I_n})(\alpha\circ \omega)=\rho({\cal A})\mathrm{v}_1^{\prime},
	\end{equation}
	where $\alpha=\mathrm{vec}({\cal A})$ and $\omega=\mathrm{vec}(\Omega)$. Denote $\mathrm{M}=({\mathrm{v}_1^{\prime}}^\mathrm{T}\otimes \mathrm{I_n})$. Then, $\mathrm{M}(\alpha\circ \omega)$ can be reduced to $\mathrm{{\hat M}}{\hat\omega}$, where ${\hat\omega}$ and $\mathrm{\hat M}$ are obtained by eliminating component $i$ and column $i$ of ${\hat\omega}$ and $\mathrm{M}$, respectively, if ${\hat\omega}_i=0$. Thus, \eqref{M1} is deduced as
	\begin{equation}
		\label{sle}
		\mathrm{{\hat M}}{\hat \omega}=\rho({\cal A})\mathrm{v}_1^{\prime}.
	\end{equation}
	In other words, each solution of the equation \eqref{sle} with ${\hat \omega}$ as the unknown variable can be utilized to expanded into a communication weight matrix. The following proves that the linear equation \eqref{sle} has a solution.
	\par Tracing back to \eqref{M1}, according to Theorem \ref{THPF}, $\mathrm{v}_1^{\prime}>0$, which results in each column of matrix $\mathrm{M}$ having only one positive element. So, $\mathrm{\hat M}$ conforms to this characteristic. Furthermore, it is easy to verify that $\mathrm{rank}(\mathrm{{\hat M}})=\mathrm{rank}([\mathrm{{\hat M}}\,\rho({\cal A})\mathrm{v}_1^{\prime}])\le \mathrm{m}$, where $\mathrm{rank}(\bullet)$ denotes taking the rank of matrix $\bullet$, $\mathrm{m}$ is the number of communication links. Thus, the system of linear equations \eqref{sle} must have a solution.
	\par Besides, thanks to $\rho({\cal A})\mathrm{v}_1^{\prime}>0$, the solutions of the system of linear equations \eqref{sle} are positive according to Lemma \ref{lemeq}. 
	\par So far, Lemma \ref{lemwt} is proved.$\hfill\blacksquare$
\end{proof} 
\par Given that more than one solution to \eqref{sle} can be found by Lemma \ref{lemwt}, in this paper, a solution that satisfies the minimum 2-norm is adopted to continue the subsequent work. In this sense,
\begin{equation}
	\label{Wm}
	{\hat\omega}=\mathrm{{\hat M}}^+\rho({\cal A})\mathrm{v}_1^{\prime},
\end{equation}
where $\mathrm{{\hat M}}^+$ is the Moore–Penrose inverse of $\mathrm{{\hat M}}$.
\subsection{Design of Secondary Voltage Controller}
\begin{figure}
	\centering
	\includegraphics[width=6cm]{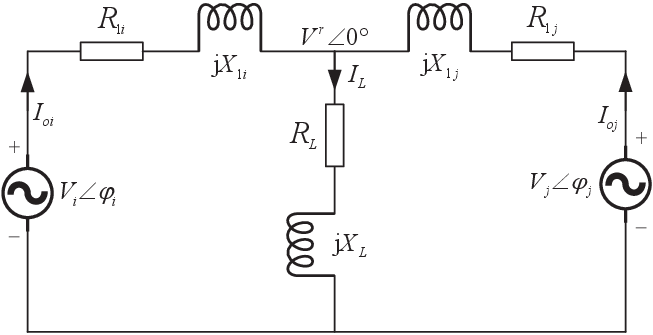}
	\caption{An equivalent circuit diagram of two parallel inverters.}
	\label{Eqc}
\end{figure}
\par According to in Fig. \ref{Eqc}, the voltage drop of DG $i$ on line impedance is
$$\Delta V_i=V_i-V^\mathrm{r}=\frac{P_i\mathrm{R}_{1i}+Q_i\mathrm{X}_{1i}}{V^\mathrm{r}}\neq 0.$$
In low-voltage microgrids, it is more practical to consider a general configuration, i.e., $\mathrm{K}_i^P\neq \mathrm{R}_{1i}$ and $\mathrm{K}_i^Q\neq \mathrm{X}_{1i}$. Denote $\mathrm{K}_i^P=\frac{1}{\delta_i^P}\mathrm{R}_{1i}$ and $K_i^Q=\frac{1}{\delta_i^Q}X_{1i}$, where $\delta_i^P$ and $\delta_i^Q$ are positive. Thus, in order to achieve power sharing, the voltage drop difference between two parallel inverters on the line is
$$\Delta V=\Delta V_i-\Delta V_j=(\delta_i^P-\delta_j^P)\frac{\mathrm{K}_i^PP_i}{V^\mathrm{r}}+(\delta_i^Q-\delta_j^Q)\frac{\mathrm{K}_i^QQ_i}{V^\mathrm{r}}\neq 0.$$
This indicates that in this case, $V_i=V_j$ is not valid for $\forall i,j\in\{1,2,\cdots,\mathrm{n}\}$, and the same goes for $V_i=V^\mathrm{r}$.
\begin{problem}
	\label{P3}
	For a power network containing $\mathrm{n}$ BESSs, it is always desirable to allocate reactive power to each according to its capacity, mathematically,
	$\lim\limits_{t\to+\infty}(\mathrm{K}_i^QQ_i-\mathrm{K}_j^QQ_j)=0,$
	where $\frac{1}{\mathrm{K}_i^Q}$ denotes the capacity of BESS $i$. Besides, in order to ensure power sharing, the difference between the output voltage of each BESS and the PCC needs to be no less than 0, i.e.,
	$\lim\limits_{t\to+\infty}(V_i-{\kappa}_iV^\mathrm{r})=0,$
	where ${\kappa}_iV^\mathrm{r}$ denotes the expected output voltage amplitude of BESS $i$.
\end{problem}
\par In the existing schemes, the output voltage of each BESS is consistent with the voltage at PCC, which makes it impossible to achieve power sharing. Then, to ensure the desired power flow, a novel scheme is designed as follows,
\begin{subequations}
	\begin{equation}
		\label{Vu}
		u_i^V=\mathrm{k}_1(\frac{1}{\rho({\cal A})}\sum_{j=1}^{\mathrm{n}}\mathrm{z}_{ij}\omega_{ij}^V{V}_j-{V}_i)+\mathrm{k}_1\mathrm{b}_{ii}(V_i-{\kappa}_iV^\mathrm{r}),
	\end{equation}
	\begin{equation}
		\label{Q}
		u_i^Q=\mathrm{k}_1(\frac{1}{\rho({\cal A})}\sum_{j=1}^{\mathrm{n}}\mathrm{z}_{ij}\omega_{ij}^QQ_j-Q_i),
	\end{equation}
\end{subequations}
where $\Omega^V=[\omega_{ij}^V]$ and $\Omega^Q=[\omega_{ij}^Q]$ are the communication weight matrices calculated by Lemma \ref{lemwt}, $\rho({\cal A})$ is the spectral radius of $\cal A$, $\mathrm{k}_1>0$ is a control gain.
\begin{theorem}
	By Lemma \ref{lemwt}, one can always find a communication weight matrices $\Omega^V=[\omega_{ij}^V]$ and $\Omega^Q=[\omega_{ij}^Q]$ that enables Problem \ref{P3} to be solved under the drive of \eqref{Vu} and \eqref{Q}.
\end{theorem}
\begin{proof}
	Define a control error vector as $e^Q=(\mathrm{I_n}-\mathrm{K}_Q\mathrm{K}^\mathrm{T}_Q)Q$, where $\mathrm{K}_Q=\mathrm{col}(\mathrm{K}_1^Q,\cdots,\mathrm{K}_\mathrm{n}^Q)$, $\mathrm{I_n}$ is an $\mathrm{n}$-dimensional identity matrix. In view of $(\mathrm{I}-\mathrm{K}_Q\mathrm{K}^\mathrm{T}_Q)\mathrm{K}_Q=0_\mathrm{n}$, the above can be yielded to
	$$\dot e^Q=\mathrm{k}_1(\mathrm{I_n}-\mathrm{K}_Q\mathrm{K}^\mathrm{T}_Q)Ae^Q,$$
	where $A=\frac{1}{\rho({\cal A})}({\cal A}\circ \Omega^Q)-\mathrm{I_n}$. According to Assumption \ref{assu1}, ${\cal A}\circ \Omega^Q$ is non-negative and irreducible. Furthermore, $-A$ is non-negative and irreducible. Denote $\Lambda_A^Q=\mathrm{diag}(\rho({\cal A}),\Lambda_2^Q,\cdots,\Lambda_n^Q)$ as an ordered Jordan canonical form of $A$, where $\Lambda_i$ is an Jordan block of $A$, $T^{-1}_QAT_Q=\Lambda_A$. Let ${\tilde e}^Q=T^{-1}_Qe^Q$. Then, one can get
	$$\begin{aligned}
		\dot{\tilde e}^Q&=\mathrm{k}_1T^{-1}_Q(\mathrm{I_n}-\mathrm{K}_Q\mathrm{K^T}_Q)AT_Q{\tilde e}^Q\\
		&=\mathrm{k}_1(\mathrm{I_n}-T^{-1}_Q\mathrm{K}_Q\mathrm{K^T}_QT_Q)\Lambda_A^Q{\tilde e}^Q.
	\end{aligned}$$
	It is noticed that $\mathrm{K}_Q=T_Q[1,0,\cdots,0]^\mathrm{T}$. Thus, 
	\begin{equation}
		\label{Rr}
		\mathrm{I_n}-T^{-1}_Q\mathrm{K}_Q\mathrm{K^T}_QT_Q=\left[\begin{matrix}
			0&-\nu\\
			0_{\mathrm{n}-1}&\mathrm{I_{n-1}}
		\end{matrix}\right],
	\end{equation}
	where $T_Q=[v_1v_2,\cdots,v_n]$, $\nu=[v_1^\mathrm{T}v_2,\cdots,v_1^\mathrm{T}v_\mathrm{n}]$. Denote $\Lambda^Q=\mathrm{diag}(\Lambda_2^Q,\cdots,\Lambda_\mathrm{n}^Q)$. The above can be deduced as
	$$\dot{\tilde e}^Q=\mathrm{k}_1\left[\begin{matrix}
		0&-\nu\Lambda^Q\\
		0_{n-1}&\Lambda^Q
	\end{matrix}\right]{\tilde e}^Q\textcolor{blue}{.}$$
	It can be easily observed that the eigenvalues of $\left[\begin{matrix}
		0&-\nu\Lambda^Q\\
		0_{n-1}&\Lambda^Q
	\end{matrix}\right]$ and $A$ are same and all of them have negative real parts, which causes ${\tilde e}^Q$ to asymptotically converge to 0.
	\par Define a control error vector $e^V=\kappa V^\mathrm{r}-V$. Find the derivative of $e^V$ as
	$$\dot e^V=\dot V=u_i^V=\mathrm{k}_1(A-\mathrm{B})e^V.$$
	The following demonstrates that $e^V_i$ for $\forall i\in\{1,2,\cdots,\mathrm{n}\}$ can converge to 0 by proving that $(A-\mathrm{B})$ is Hurwitz.
	\par Define two matrices to be used first as follows,
	$${\tilde A}=\rho({\cal A})\mathrm{I_n}-({\cal A}\circ\Omega)=-\rho({\cal A})A,$$
	$$F={\tilde A}+\rho({\cal A})\mathrm{bb^T}=-\rho({\cal A})(A-\mathrm{B}).$$
	Due to ${\cal A}\circ\Omega^V$ being an irreducible non negative matrix, ${\tilde A}$ has positive right and left eigenvectors related to zero eigenvalues. Therefore, according to Lemma \ref{nons}, matrix $F$ has all positive leading mirrors, which implies that $F$ is a irreducible M-matrix, i.e., $F>0$. According to Lemma \ref{MH}, $(A-\mathrm{B})$ is a Hurwitz matrix. That is, it can be asserted that $e^V_i$ for $i\in\{1,2,\cdots,\mathrm{n}\}$ asymptotically converges to 0.
	\par So far, Problem \ref{P3} is solved.
	$\hfill\blacksquare$
\end{proof}
\begin{remark}
	Unlike the previous solution designed in \cite{9531546,QIN2024110255,MO202448,9989432,9676705,10092457,10032563,LIU2023108679,9733168,10266581}, the application of \eqref{Vu} and \eqref{Q} can achieve voltage leader-following multiconsensus and reactive power leaderless multiconsensus with the preset steady states, thereby solving Problem \ref{P3}. The required communication weights can be calculated by the designed scheme according to specific requirements, and then induce the desired voltage and reactive power states using control protocols \eqref{Vu} and \eqref{Q}.
\end{remark}
\subsection{Design of A Controller to Achieve SoC Balance and Power Sharing}
\par Define the unit SoC decrease as
\begin{equation}
	E_{\mathrm{d},i}=\frac{E_0-E_i}{\mathrm{K}_i^P},
\end{equation}
where $E_{0i}$ is the SoC initial state of BESS $i$, $E_i$ is the SoC of BESS $i$. Thus, a linear second-order model as follows to describe the dynamics of BESS $i$, i.e.,
\begin{equation}
	\label{Ed}
	\left\{\begin{matrix}
		\dot E_{\mathrm{d},i}=P_i,\\
		\dot P_i=u_i.
	\end{matrix}\right.\
\end{equation}
\begin{problem}\label{P2}
	Under Assumption \ref{assu1}, in order to achieve SoC balance, i.e., $\lim\limits_{t\to+\infty}(E_i-E_j)=0$, for $i,j\in\{1,2,\cdots,\mathrm{n}\}$, the decrease in SoC level needs to achieve consensus, i.e.,
	$\lim\limits_{t\to+\infty}(\mathrm{K}_i^PE_{\mathrm{d},i}-\mathrm{K}_j^PE_{\mathrm{d},i})=0.$
	Besides, proportional power sharing can be achieved according to capacity, i.e.,
	$\lim\limits_{t\to+\infty}(\mathrm{K}_i^PP_i-\mathrm{K}_j^PP_j)=0.$
\end{problem}
\par The designed controller to solve Problem \ref{P2} is
\begin{equation}
	\label{Eu}
	u_i=\mathrm{k}_1(\sum_{j=1}^{\mathrm{n}}a_{ij}\omega_{ij}^PP_j-P_i)+\mathrm{k}_2(\sum_{j=1}^{\mathrm{n}}\mathrm{z}_{ij}\omega_{ij}^PE_{\mathrm{d},j}-E_{\mathrm{d},i}),
\end{equation}
where $\Omega^P=[\omega_{ij}^P]$ is the weight matrix obtained by Lemma \ref{lemwt}, $\mathrm{k}_1$ and $\mathrm{k}_2$ are control gains to be designed.
\begin{theorem}
	Consider a power network with $\mathrm{n}$ BESSs that follows Assumption \ref{assu1}, each of which follows the dynamics shown in \eqref{Ed}. Under the drive of controller \eqref{Eu}, Problem \ref{P2} can be solved using the communication weights calculated by Lemma \ref{lemwt}.
\end{theorem}
\begin{proof}
	Define $x_i=\mathrm{col}(E_{\mathrm{d},i},P_i)$. Thus,
	$$\dot x_i=\mathrm{A}_xx_i+\mathrm{B}_xu_i,$$
	where $\mathrm{A}_x=\left[\begin{matrix}
		0&1\\
		0&0
	\end{matrix}\right]$, $\mathrm{B}_x=\left[\begin{matrix}
	0\\
	1
	\end{matrix}\right]$. Define a control error vector as $e^x=(\mathrm{I_n}-\mathrm{K}_P\mathrm{K}^\mathrm{T}_P)\otimes \mathrm{B}{\cal K}x$, where $\mathrm{K}_P=\mathrm{col}(\mathrm{K}_1^P,\cdots,\mathrm{K}_\mathrm{n}^P)$. Thus,
	$$\dot e^x=(\mathrm{I_n}-\mathrm{K}_P\mathrm{K^T}_P)\otimes \mathrm{A}_xe^x+(\mathrm{I_n}-\mathrm{K}_P\mathrm{K^T}_P)A\otimes \mathrm{B}_x{\cal K}e^x.$$
	Letting ${\tilde e}^x=(T^{-1}_x\otimes \mathrm{I}_2)e^x$, where $\mathrm{K}_P=T_x[1,0,\cdots,0]^\mathrm{T}$, one can get
	$$\begin{aligned}
		\dot{\tilde e}^x=&T^{-1}_x(\mathrm{I_n}-\mathrm{K}_P\mathrm{K}^\mathrm{T}_P)T_x\otimes \mathrm{A}_x{\tilde e}^x\\
		&+T^{-1}_x(\mathrm{I_n}-\mathrm{K}_P\mathrm{K}^\mathrm{T}_P)AT_x\otimes \mathrm{B}{\cal K}{\tilde e}^x\\
		=&(\mathrm{I_n}-T^{-1}_x\mathrm{K}_P\mathrm{K}^\mathrm{T}_PT_x)\otimes \mathrm{A}_x{\tilde e}^x\\
		&+(\mathrm{I_n}-T^{-1}_x\mathrm{K}_P\mathrm{K}^\mathrm{T}_PT_x)\Lambda_A^E\otimes \mathrm{B}{\cal K}{\tilde e}^x,
	\end{aligned}$$
	where $\Lambda_A^E=\mathrm{diag}(\rho({\cal A}),\Lambda_2^E,\cdots,\Lambda_\mathrm{n}^E)$. Define $\mathrm{Y}=(\mathrm{I_n}-T^{-1}_x\mathrm{K}_P\mathrm{K}^\mathrm{T}_PT_x)\otimes \mathrm{A}_x+(\mathrm{I_n}-T^{-1}_x\mathrm{K}_P\mathrm{K}^\mathrm{T}_PT_x)\Lambda^E\otimes \mathrm{B}{\cal K}$, where $\Lambda^E=\mathrm{diag}(\Lambda_2^E,\cdots,\Lambda_\mathrm{n}^E)$. Due to \eqref{Rr}, the matrix $Y$ is a lower triangle with block $j$ as, $j\in\{2,\cdots,\mathrm{n}\}$,
	$$\mathrm{Y}_j=\left[\begin{matrix}
		0&1\\
		\mathrm{k}_1\lambda_j&\mathrm{k}_2\lambda_j
	\end{matrix}\right],$$
	where $\lambda_j$ is an eigenvalue of $\Lambda^E$, and which has a characteristic polynomial as follows
	\begin{equation}
		\label{poly}
		s^2-\mathrm{k}_2\lambda_js-\mathrm{k}_1\lambda_j.
	\end{equation}
	\par From previous analysis, it can be seen that $\mathrm{Y}_1$ is related to the states of multiconsensus, while other blocks determine the stability of the system. Then, it is necessary to find $\mathrm{k}_1$ and $\mathrm{k}_2$ to ensure each block $\mathrm{Y}_j$, $j\in\{2,\cdots,\mathrm{n}\}$, is Hurwitz. The complex form of eigenvalue $j$ is denoted as $\lambda_j=\mathrm{a}_j+i\mathrm{b}_j$, where $i$ is the imaginary unit. Thus, \eqref{poly} can be rewritten as
	$$s^2-\mathrm{k}_2(\mathrm{a}_j+i\mathrm{b}_j)s-k_1(\mathrm{a}_j+i\mathrm{b}_j).$$
	According to the conclusion in \cite{ST1993}, the roots of this characteristic polynomial all have negative real parts if and only if
	\begin{equation}
		\left\{\begin{aligned}
			&\mathrm{k}_2\mathrm{a}_j<0, j\in\{2,\cdots,\mathrm{n}\}\\
			&\mathrm{k_1k_2a}_j^2+\mathrm{k_1^2b}_j^2+\mathrm{k_1k_2^2a}_j\mathrm{b}_j^2<0.
		\end{aligned}\right.
	\end{equation}
	\par So far, Problem \ref{P2} is solved.$\hfill\blacksquare$
\end{proof}
\section{Case Study}
\begin{figure}
	\centering
	\includegraphics[width=7cm]{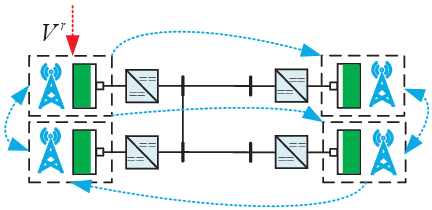}
	\caption{A 4-bus system and equipped communication network.}
	\label{57bus}
\end{figure}
\begin{figure}
	\centering
	\includegraphics[width=7cm]{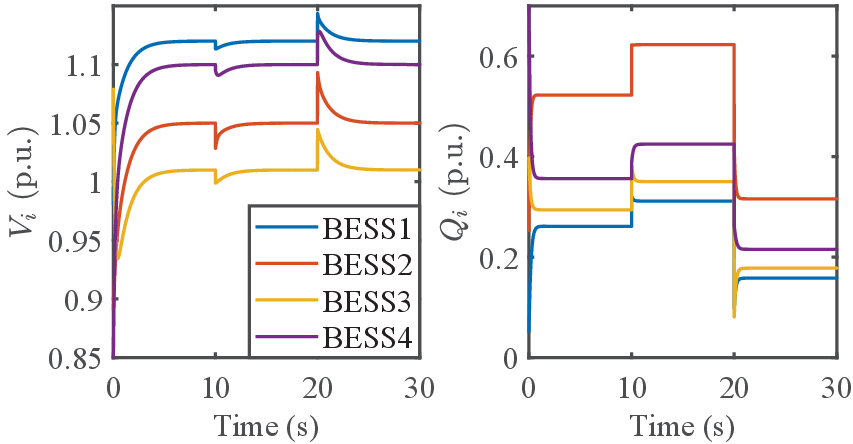}
	\caption{Evolution of voltage and reactive power of each BESS.}
	\label{voltage}
\end{figure}
To verify whether the designed method can solve Problems \ref{P3} and \ref{P2}, a low-voltage BESS network equipped by a communication topology shown in Fig. \ref{57bus} is selected for simulation, where $\kappa=\mathrm{col}(1.12,1.05,1.01,1.1)$, $\mathrm{K}_i^Q=\mathrm{col}(0.9,0.45,0.8,0.66)$, and $\mathrm{K}_i^P=\mathrm{col}(0.9,0.45,0.8,0.66)$. Meanwhile, some events are considered here. That is, there are a load increase $\mathrm{col}(0.06,0.1,0.09,0.03)$ and a load decrease $\mathrm{col}(0.11,0.1,0.14,0.08)$ at all load nodes in $\mathrm{p.u.}$.
\par From the simulation results in Fig. \ref{voltage}, it can be seen that the output voltage of each BESS can converge to the preset value. Meanwhile, reactive power can be evenly distributed according to its capacity. It needs to be emphasized again that using the designed scheme, in this simulation, only the agent managing BESS 1 can obtain its unique expected voltage value. Moreover, compared to the past, the voltage droop coefficient is no longer transmitted, which means it is no longer directly utilized by neighboring agents.
\par From the simulation results in Fig. \ref{SoCP}, it can be seen that the output power of all BESSs can be evenly distributed according to capacity, and SoC can also achieve balance. It is worth mentioning that SoC does not necessarily have to be sent to its neighbors, which is very different from the past. As long as the communication weights are designed properly, active power can still be evenly distributed according to capacity even without transmitting capacity information, which is another unique feature of this article compared to the past.
\begin{figure}
	\centering
	\includegraphics[width=7cm]{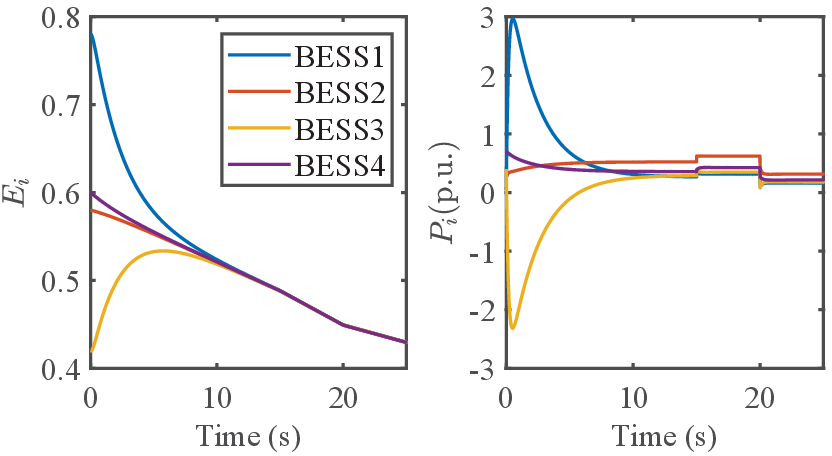}
	\caption{Evolution of SoC and active power of each BESS.}
	\label{SoCP}
\end{figure}
\begin{remark}
		It should be noted that the design of the distributed consensus scheme in this article is not intended to achieve better dynamic and steady-state performance than existing works, but to provide a more flexible and efficient solution for secondary voltage control and SoC balance. When implementing this scheme, capacity information/droop coefficient does not need to be transmitted, thereby reducing the pressure on the communication network. Through the designed scheme, the output voltage of each BESS can be precisely controlled. This is different from the result in \cite{10266581} because it does not design the secondary voltage control scheme from the perspective of prescribed steady state, and does not reveal how the voltage steady state is induced by underlying communication network, which is the core advantage of this work. In addition, although voltage multiconsensus control can be achieved under this scheme observed from the simulation results, voltage consensus \cite{QIN2024110255}/containment \cite{10266581} control still can be achieved as long as the communication weights are designed reasonably. The above analysis is summarized into a table as shown in Table I.
		\begin{table}[h]\label{tab0}
			\caption{Comparison with previous methods}
				\centering
				\begin{tabular}{cccc}
					\toprule
					Schemes &\makecell[c]{The one in\\this paper}&\makecell[c]{The one\\in \cite{QIN2024110255}}&\makecell[c]{The one\\in \cite{10266581}}\\
					\midrule
					Capacity information & $\times$ & $\checkmark$ & $\checkmark$ \\
					Consensus control & $\checkmark$ & $\checkmark$ & $\times$\\
					Multiconsensus control & $\checkmark$ &$\times$ & $\times$\\
					Containment control& $\checkmark$ & $\times$ & $\checkmark$\\
					Prescribed steady state&$\checkmark$ & $\checkmark$ & $\times$\\
					\bottomrule
			\end{tabular}
		\end{table}
\end{remark}
\section{Conclusion}
In this article, the distributed secondary voltage control and SoC balancing control of BESSs based on eigenvector centrality are investigated. Compared with existing results, by designing communication weights over a directed and connected graph, proportional sharing of active/reactive power and SoC balancing without the need to transmit capacity information are achieved, and similarly, secondary voltage control with preset steady state and single leader is achieved. It can be foreseen that this scheme can be expanded into finite/fixed/preset time schemes in the near future, and further combined with event triggering mechanisms to develop more efficient schemes.
\par However, in order to implement the currently designed solution, an aggregator is needed to master the communication topology and generate a reference voltage for DG as the leader of MASs. These slightly centralized methods need further improvement.
\bibliography{BESS_reference.bib}
\bibliographystyle{IEEEtran}
\end{document}